\documentclass[10pt,twocolumn,aps,prl,showpacs,superscriptaddress,floatfix]{revtex4}

\usepackage{epsfig}
\usepackage{amsmath, amssymb, amsfonts, mathrsfs}
\usepackage{amsthm}
\usepackage{graphicx}
\usepackage{color}
\usepackage{hyperref}
\usepackage{bbm}
\usepackage{soul}

\newcommand{\ket}[1]{\vert{#1}\rangle} 
\newcommand{\bra}[1]{\langle{#1}\vert} 
 
\newcommand{\proj}[1]{\ket{#1}\!\bra{#1}}
\newcommand{\mean}[1]{\langle #1 \rangle}

\DeclareMathOperator{\Tr}{Tr}

\renewcommand\Im{\operatorname{Im }}

\newcommand{\beq}{\begin{equation}}
\newcommand{\eeq}{\end{equation}}

\newcommand{\be}{\begin{equation}}
\newcommand{\ee}{\end{equation}}

\newcommand{\ben}{\begin{eqnarray}}
\newcommand{\een}{\end{eqnarray}}
\usepackage{amsmath}

\newtheorem{theorem}{Theorem}

\begin{document}

\title{Heisenberg-Weyl Observables: Bloch vectors in phase space}

\author{Ali Asadian}
\affiliation{{Naturwissenschaftlich-Technische Fakult\"at, Universit{\"a}t Siegen, Walter-Flex-Str.~3, D-57068 Siegen, Germany}}
\author{Paul Erker}
\affiliation{Universitat Autonoma de Barcelona, 08193 Bellaterra, Barcelona, Spain}
\affiliation{Faculty of Informatics, Universit\`{a} della Svizzera italiana, Via G. Buffi 13, 6900 Lugano, Switzerland}
\affiliation{Facolt\`{a} indipendente di Gandria, Lunga scala, 6978 Gandria, Switzerland}
\author{Marcus Huber}
\affiliation{Group of Applied Physics, University of Geneva, 1211 Geneva 4, Switzerland}
\affiliation{Universitat Autonoma de Barcelona, 08193 Bellaterra, Barcelona, Spain}
\affiliation{ICFO-Institut de Ciencies Fotoniques, 08860 Castelldefels, Barcelona, Spain}
\affiliation{Institute for Quantum Optics and Quantum Information (IQOQI), Austrian Academy of Sciences, 
A-1090 Vienna, Austria}
\author{Claude Kl\"ockl}
\affiliation{Universitat Autonoma de Barcelona, 08193 Bellaterra, Barcelona, Spain}

\date{\today}

\begin{abstract}
We introduce a Hermitian generalization of Pauli matrices to higher dimensions which is based on Heisenberg-Weyl operators. The complete set of Heisenberg-Weyl observables allows us to identify a real-valued Bloch vector for an arbitrary density operator in discrete phase space, with a smooth transition to infinite dimensions. Furthermore, we derive bounds on the sum of expectation values of any set of anti-commuting observables. Such bounds can be used in entanglement detection and we show that Heisenberg-Weyl observables provide a first non-trivial example beyond the dichotomic case.
\end{abstract}

\maketitle 

\emph{Introduction.} The Bloch representation is a cornerstone of analyzing the characteristics of quantum systems. It was first introduced for two-level systems by Bloch \cite{Bloch} and has since been used in a wide variety of settings (for comprehensive reviews consult \cite{Fano,Krammer, Petruccione}). It is usually defined via a decomposition of the density matrix into a complete operator basis. Defining quantum states via the expectation
values of a complete set of measurements gives a
very practical account of their properties. Apart from being intuitive this approach gives convenient solutions for Hamiltonian evolutions (see, e.g., \cite{Eberly}) and has found many applications in entanglement theory \cite{Vicente1,Vicente2,Paterek,Laskowski,Schwemmer,Kloeckl, EntangledBloch}.

However, there is not a unique Bloch decomposition of a given quantum state; this fact may favor a particular representation over another for certain tasks. 
The canonical choice for a complete basis of observables is usually given by the so-called generalized Gell-Mann matrices, generators of the special unitary group [SU($d$)].  Being a natural choice for higher-dimensional spin representations they have been extensively used in parametrizations of corresponding density matrices \cite{Petruccione, Spengler} and in entanglement detection. Other choices, such as the Heisenberg-Weyl (HW) operators, and the non-Hermitian generalization of the 1/2-spin Pauli operators, have also been explored \cite{Krammer, Massar, Ryo, Asadian, Cotfas, Vourdas}. While having some convenient properties, they are unitary, but not Hermitian matrices. Thus the associated Bloch vector itself has imaginary entries that do not correspond to expectation values of physical observables. This makes both the theoretical description and experimental realization more cumbersome and therefore requires more effort to identify the relevant parameters.

In this article we introduce a Hermitian Bloch-basis derived from HW-operators. It conveniently combines multiple desirable properties of Bloch vector parametrizations, allowing a smooth transition to the infinite dimensional limit. We first explore properties such as (anti-) commutativity. We then proceed with the derivation of an inequality which bounds sums of anti-commuting observables with which we show in exemplary cases how one can construct powerful criteria for entanglement detection in this new basis. Finally we present a scheme for practical experimental acquisition through a Ramsey-type measurement.

\emph{Phase-space displacements.} We start our analysis with a short reiteration of the HW-operator basis. The operators $Z=e^{i2\pi Q/d}$ and  $X=e^{-i2\pi P/d}$ describe generalized Pauli ``phase" and ``shift" operators with effect $X\ket{j}=\ket{j+1 \  {\rm mod}  \ d }$ and $Z\ket{j}=e^{i2\pi j/d}\ket{j}$. $Q$ and $P$  are discrete position and momentum operators, respectively, describing a $d\times d$ grid.  $X$ and $Z$  operators are non-commutative in general and obey the relation
\begin{equation}
\label{ZXcomm}
Z^lX^m=X^mZ^l e^{i2\pi l m/d}.
\end{equation}

The unitaries corresponding to discrete phase-space
displacements for $d$-level systems are defined as
\begin{equation}
\label{eq:Dlm}
\mathcal{D}(l,m)=Z^lX^m e^{-i\pi lm/d}.
\end{equation}
Displacement operators hold a number of convenient properties which will be particularly
useful in our analysis. Principal among these is the completeness of displacement operators. That
is, they form a complete non-Hermitian orthogonal basis satisfying the orthogonality condition
\begin{equation}
\text{Tr}\{\mathcal{D}(l,m)\mathcal{D}^\dag(l',m')\}=d\delta_{l,l'}\delta_{m,m'}.
\end{equation}
Therefore, any bounded operator, including density operators $\rho$, can be decomposed into
\begin{equation}
\rho=\dfrac{1}{d}\sum_{l,m=0}^{d-1}\text{Tr}\{\rho\mathcal{D}(l,m)\}\mathcal{D}^\dag(l,m)\equiv \dfrac{1}{d}(\mathbbm{1}+\vec{\xi}\cdot\vec{\mathcal{D}}^\dag),
\end{equation}
from which the Bloch representation is already apparent.

In this formulation, however, the Bloch vector components, $\xi_{lm}=\text{Tr}\{\rho\mathcal{D}(l,m)\}$ are generally complex as the displacement operators are not Hermitian. Therefore we need to determine $(d^2-1)$ complex parameters of the Bloch vector, $\vec{\xi}$,  to fully characterize the density operator.  The obvious question here is, can we find a minimal complete set of $d^2-1$ Hermitian operators whose expectation values with respect to the density operator are sufficient to fully characterize the state. In the following we develop a basis which has the above property.

\emph{HW observable basis.}  The standard Hermitian generalization of Pauli operators used in quantum information theory are called generalized Gell-Mann matrices (GGM) \cite{Krammer}. Alternatively, here we are aiming to identify a minimal and complete set of Hermitian operators constructed from the HW operators, $\mathcal{D}(l,m)$. 

\noindent
\textbf{Ansatz}. Our attempt begins by making an Ansatz solution of the form
\begin{equation}
\label{eq:Q}
\mathcal{Q}(l,m)=\chi\mathcal{D}(l,m)+\chi^*\mathcal{D}^\dag(l,m),
\end{equation} 
which is Hermitian  by construction. In order to form a valid basis of observables
this Ansatz has to satisfy the orthogonality condition
\begin{equation}
\label{eq:orth}
\text{Tr}\{\mathcal{Q}(l,m)\mathcal{Q}(l',m')\}=d\delta_{l,l'}\delta_{m,m'} .
\end{equation}
The above condition is satisfied only for the choice (see Appendix Sec. I)
\begin{equation}
\chi=\dfrac{(1\pm i)}{2}.
\end{equation} 
Therefore, we establish $d^2-1$ orthogonal and traceless observables which are linearly independent. The $d^2-1$ observables plus identity matrix, $\mathcal{Q}(0,0)=\mathbbm{1}_d$  form a basis acting on a $d$ dimensional Hilbert space, and thus provide a Bloch representation of an arbitrary state. This enables us to decompose any density operator in terms of HW observables of the form
\begin{equation}
\rho=\dfrac{1}{d}\sum_{l,m=0}^{d-1} \mean{\mathcal{Q}(l,m)}\mathcal{Q}(l,m).
\end{equation}
 
 This basis simply reduces to the Pauli matrices for $d=2$. We henceforth refer to its elements as ``Heisenberg-Weyl observables". They have distinct properties from those of the GGM matrices which will turn out advantageous in some tasks. First, HW observables in contrast with GGM generically have full-rank, making them more efficient in sparsely characterizing states with a lot of coherence. Secondly, for a suitable parametrization, strict pairwise (anti-)commutativity relations can be obtained which can be applied to entanglement detection and will be demonstrated later . 

For continuous variable systems it is of great practical importance to find operational discretizations for processing quantum information \cite{kett1,krenn}. Notably, HW observables can be systematically extended to the continuous limit of infinite dimensional systems, holding analogous properties and with this also all corresponding results can be extended in this limit, thus providing a natural path towards a discretization of continuous variable systems.

First, let us introduce the compact notation
\begin{equation}
\alpha\equiv\sqrt{\dfrac{\pi}{d}}(m+il)
\end{equation}
known as the displacement amplitude in discrete phase space. One can also think of $\alpha$ as a real vector in a two dimensional space $\alpha:=(\alpha_R, \alpha_I)$.  We can now write $\mathcal{Q}(\alpha):=\mathcal{Q}(\sqrt{\frac{d}{\pi}}\alpha_I,\sqrt{\frac{d}{\pi}}\alpha_R)$ and thus
\begin{align}
\rho=\dfrac{1}{d}\Big(\mathbbm{1}+\sum_{\alpha\in \mathcal{S}}\mean{\mathcal{Q}(\alpha)}\mathcal{Q}(\alpha)\Big)\,,
\end{align}
where we used $\mathcal{S}:=\{\alpha: \alpha_I=\sqrt{\frac{\pi}{d}}l,\alpha_R=\sqrt{\frac{\pi}{d}}m\}$.

To facilitate a smooth transition to infinite dimensions one can consider $\hat x=Q\sqrt{2\pi/d}$ and $\hat p=P\sqrt{2\pi/d}$ as the position and momentum operators. In this case, $X^m\equiv e^{-ix\hat p}$ indicates position displacement by $x=m\sqrt{2\pi/d}$. Similarly $Z^l\equiv e^{ip \hat x}$ displaces the momentum by $p=l\sqrt{2\pi/d}$. Therefore Eq. \eqref{eq:Dlm} can be rewritten  
\begin{equation}
\label{eq:Dxp}
\mathcal{D}(p,x)\equiv e^{ip\hat x} e^{-ix \hat p} e^{-ix p/2}.
\end{equation}
In the limit $d\rightarrow\infty$ we recover the Heisenberg commutation relation for position and momentum of a continuous variable system, $[\hat x,\hat p]=i$.  We can then use the special form of the Baker-Campbell-Hausdorff formula for exponential operators, i.e., $e^{A+B}=e^A e^Be^{-[A,B]/2}$ where $[A,[A,B]]=0=[B,[A,B]$. Thus, the definition \eqref{eq:Dxp} can be written   $\mathcal{D}(p,x)=e^{ip \hat x-ix \hat p}$, the form of which is valid only in the infinite dimensional limit. An equivalent reformulation of this is $\mathcal{D}(\alpha)=e^{\alpha a^\dag-\alpha^* a}$ with orthogonality condition $\text{Tr}\{\mathcal{D}^\dag(\alpha)\mathcal{D}(\alpha')\}=\pi\delta^2(\alpha-\alpha')$ where $a^\dag(a)$ denotes creation (annihilation) operators of a bosonic mode and $\alpha$ is the displacement amplitude. Therefore, the continuous analog of \eqref{eq:orth} is
\begin{equation}
\text{Tr}\{\mathcal{Q}(\alpha)\mathcal{Q}(\alpha')\}=\pi\delta^2(\alpha-\alpha').
\end{equation}
The discrete-continuous transition is therefore identified by the replacement $\dfrac{1}{d}\sum_\alpha\rightarrow \dfrac{1}{\pi}\int d^2 \alpha$. This shows that our observable basis developed for discrete systems can be systematically extended to continuous variable systems.

\emph{Anti-commutativity}. The very feature of the Pauli operators is the fact that all of them are mutually anti-commuting. Amongst other things this allows for tight bounds on uncertainty relations and can be used in entanglement detection \cite{Guehne, Vicente1,Vicente2,Paterek,Laskowski,Schwemmer,Kloeckl,ACotfried,ACotfried2}. Hence it will be interesting to analyze the (anti-)commutation relations for the HW-observables. From Eq.\eqref{ZXcomm} it follows

\begin{equation}
\mathcal{D}(\alpha)\mathcal{D}(\alpha')=e^{i2\alpha\times\alpha'}\mathcal{D}(\alpha')\mathcal{D}(\alpha)\,,
\end{equation}
obeyed by both discrete and continuous phase space displacement operators. Recall, $\Im (\alpha\alpha'^*):=\alpha\times\alpha'$.
This allows the convenient characterization of commutativity and anti-commutativity among all basis elements. From the above equation it is obvious that any pair of displacement operators satisfying
\begin{equation}
\label{eq:condition}
|\alpha\times\alpha'|=\dfrac{\pi}{2}(2n+1)\leq \dfrac{\pi(d-1)^2}{d},
\end{equation} 
 anti-commutes, and so does the corresponding HW observable pair accordingly, i.e., 
\begin{equation}
\{\mathcal{Q}(\alpha),\mathcal{Q}(\alpha')\}=0,
\end{equation}
which can be easily verified. In the discrete notation this means $|m'l-ml'|=d(2n+1)/2\leq (d-1)^2$.  An interesting observation made from this condition is that anti-commutativity between HW observables cannot be achieved strictly in odd-dimensional systems. This is because the left hand side is always integer while the right hand side can only be integer when $d$ is even.
The corresponding condition for achieving commutativity is $|\alpha \times\alpha'|=\pi n$.

\emph{Anti-commutativity bound and entanglement detection with HW observables.} 
We proceed with presenting a theorem bounding sums of squared expectation values, if the mutual anti-commutators of the observables are small.
In the Appendix section IV we present a detailed proof of the theorem. It generalizes a theorem presented in \cite{ACotfried} (later proven differently in \cite{acb1,acb2}) from dichotomic anti-commuting observables to observables with an arbitrary spectrum and non-vanishing anti-commutators.
\begin{theorem}\label{Anticommutativity Bound}
		Let $\{\lambda_{i}\}_{i \in \mathcal{I}}$ with the index set $\mathcal{I}=\{1,2,\ldots,d^{2}\}$ denote an orthonormal self-adjoint basis $\mathcal{B}$ of a $d$-dimensional Hilbert space $\mathcal{H}$ and $\mathcal{A} \subseteq \mathcal {I}$ refer to a subset of $\mathcal{B}$ such that $\frac{1}{2}\sqrt{\sum_{i \neq j \in \mathcal{A}}\langle\{\lambda_i,\lambda_j\}\rangle^2} \leq \mathcal{K}$. 
Then the corresponding Bloch vector components $c_{i}$ of any density matrix $\rho \in \mathcal{H}$ expressed in $\mathcal{B}$ as $\rho=\sum_{i \in \mathcal{A}}c_{i} \lambda_{i} + \sum_{l \in \bar{\mathcal{A}}} c_{l} \lambda_{l}$ can be bounded by
\begin{equation}
\sum_{i \in \mathcal{A}} c_{i}^{2} \leq \frac{\max_{i \in \mathcal{A}} \langle \lambda_{i}^{2} \rangle + \mathcal{K} }{[\min_{i \in \mathcal{A}} \Tr(\lambda_{i}^{2})]^{2}}
\end{equation}
\end{theorem}

One way to detect entanglement using anti-commutativity can be accomplished by first identifying a set of nonzero Bloch vector entries of a multipartite quantum state with anticommuting reductions across the partition one is interested in. That is,  we are looking for a partition $A|\bar{A}$ and a set $\beta=\beta_A\cup\beta_{\bar{A}}$ with elements $\tau_\beta=\text{Tr}(\rho \lambda_{i}\otimes \lambda_{j})$, where $i\in\beta_A$, $j\in\beta_{\bar{A}}$ and $\lambda_i(\lambda_j)$ are arbitrary observables acting on subsystem $A(\bar{A})$. The sum of moduli of these correlations can be bounded for states, which are product with respect to these partitions as 
\begin{align}
\text{Tr}(\rho_A\otimes\rho_{\bar{A}}\lambda_{i}\otimes \lambda_{j})=\text{Tr}(\rho_A \lambda_{i})\text{Tr}(\rho_{\bar{A}} \lambda_{j})\,,
\end{align}
and $|\langle u|v\rangle|\leq {\parallel u\parallel}_2 {\parallel v\parallel}_2$. Now, for the sake of convenience, let us assume that all observables $\lambda_{i}$ are anti-commuting (i.e. $\mathcal{K}=0$) and normalized [i.e. $\text{Tr}(\lambda_{i}\lambda_{i'})=d\delta_{i,i'}$]. Then we can make direct use of the anti-commutativity bound to assert that
\begin{align}
\label{eq:criterion}
\sum_{i\in\beta_A} |\text{Tr}(\rho_A \lambda_{i})|^2\leq\max\limits_{i\in\beta_A} \langle \lambda_{i}^{2} \rangle\,,
\end{align}
and analogously for $\bar{A}$. To finish we only need to point out that the original expression, a sum of moduli of expectation values, is convex in the space of density matrices and thus the validity of the inequality for product states translates to a general validity for separable states. The case of non-normalized or only partially anti-commuting observables works analogously.

For the case of HW-observables the anti-commutativity bound can be simplified to
\begin{align}
\sum_{i \in \mathcal{A}} \langle \mathcal{Q}&(\alpha_i)\rangle^{2} \leq q^2_{\rm max} + \mathcal{K}\label{hwacbound}
\end{align}
where $q_{\rm max}^2=1+\max_{n\in \mathbb{N}}\sin(4\pi n/d)$ is the maximum eigenvalue of a HW observable which is the same for any $\mathcal{Q}(\alpha_j)$ for a given dimension and $\mathcal{K}\leq\frac{1}{2}\sqrt{\sum_{i \neq j \in \mathcal{A}}\left\|\{\mathcal{Q}(\alpha_i),\mathcal{Q}(\alpha_j)\}\right\|^{2}_{\infty}}$. The case of non-vanishing anti-commutators has also been studied for dichotomic observables in the context of uncertainty relations \cite{Jed}. The central quantity $\mathcal{K}$ is proportional to the $2$-norm of the ``anti-commutator matrix'' introduced therein.  A case of particular interest is of course given by exact anti-commutativity, i.e.~ $\mathcal{K}=0$. In the Appendix (Sec. IV B) we present the proof that we have identified the maximal set of anti-commuting operators in the HW basis, i.e., that no more than three anti-commuting HW--observables exist.

\textit{Example.} To illustrate this method in an exemplary case let us turn to qudit systems with the maximally entangled state defined as 
\begin{equation}
\label{eq:quditEnt}
\ket{\phi_d}=\frac{1}{\sqrt{d}}\sum_{j=0}^{d-1}\ket{j}\ket{j}.
\end{equation}
This important class of entangled states in quantum information is a maximal resource for many tasks. Its Bloch decomposition in terms of HWOs is simply given by 
\begin{align}
 \proj{\phi_d}=\dfrac{1}{d^2}\Big(\mathbbm{1}\otimes\mathbbm{1}+\sum_{\alpha\in \mathcal{S}}\mathcal{Q}(\alpha)\otimes\mathcal{Q}(\alpha)^*\Big)
\end{align}
where  $\mathcal{Q}(\alpha)^*=\mathcal{Q}(-\alpha^*)$ denotes the complex conjugate. From above Bloch decomposition it follows that the expectation value of the correlations are all equal to $1$, and this means that measuring only three anti-commuting local observables for each party is sufficient to violate the upper bound and thus detect entanglement. The violation is obviously enhanced with three pairwise anti-commuting observables whose respective amplitudes fulfill the constraint
$|\alpha_1\times\alpha_2|=|\alpha_2\times\alpha_3|=|\alpha_3\times\alpha_1|=\pi/2(2n+1)$ yielding a general recipe for finding  three pairwise anti-commuting observables.  In this case, the criterion \eqref{eq:criterion} may written as
\begin{equation}
\sum_{i=1}^3\mean{\mathcal{Q}(\alpha_i)\otimes\mathcal{Q}(\alpha_i)^*}\stackrel{\text{DV}}\leq  {q^2_{\rm max}} \stackrel{\text{CV}}\leq 2,
\end{equation}
with respective upper bounds  on separable states for discrete (DV) and continuous variable (CV) cases.
In the continuous limit ($d\rightarrow \infty$) the above entangled state becomes a perfectly correlated Einstein-Podolski-Rosen(EPR) entangled state \cite{EPR} which is equal to an infinitely squeezed two-mode squeezed state, i.e. $\dfrac{1}{\sqrt\pi}\int_\mathbb{R} dx\ket{x}\ket{x}=\dfrac{1}{\sqrt\pi}\int_\mathbb{R} dx\ket{p}\ket{-p}$ with continuous Bloch decomposition  $\dfrac{1}{\pi^2}\int d\alpha^2\mathcal{Q}(\alpha)\otimes\mathcal{Q}(\alpha)^*$. The associated set of three pairwise anti-commuting observables in this limit is simply given by a symmetric case of three equiangular amplitudes with equal lengths $|\alpha_j|=\sqrt{\pi/\sqrt{3}}\simeq 1.34$ mutually separated by an angle $2\pi/3$. 
In comparison, the corresponding correlations with respect to the generalized Gell-Mann basis are all equal to $2/d$ \cite{Krammer}. Thus, the number of required measurements in order to detect entanglement scales with $d$ making it impractical in high dimensions. This is an example which clearly demonstrates the advantage of the HW observables in high-dimensional entanglement detection.  For more details see Appendix section IV where we give more explicit examples that the generalized anticommutativity bound can be used to detect entanglement even if the observables are nondichotomic, therefore generalizing the results obtained in \cite{ACotfried,acb1,acb2}. 

While the anticommuting elements, within one local system, are limited, tensor product bases consisting of commuting and anticommuting sets yield desired large sets. The product will be again anti-commuting as long as an odd number of factors is anticommuting. Therefore, once commuting and anti-commuting local basis elements are found, extending to larger particle numbers is possible by straightforward combinatorial calculations. 
\begin{figure}[t]
\begin{center}
\includegraphics[width=1\columnwidth]{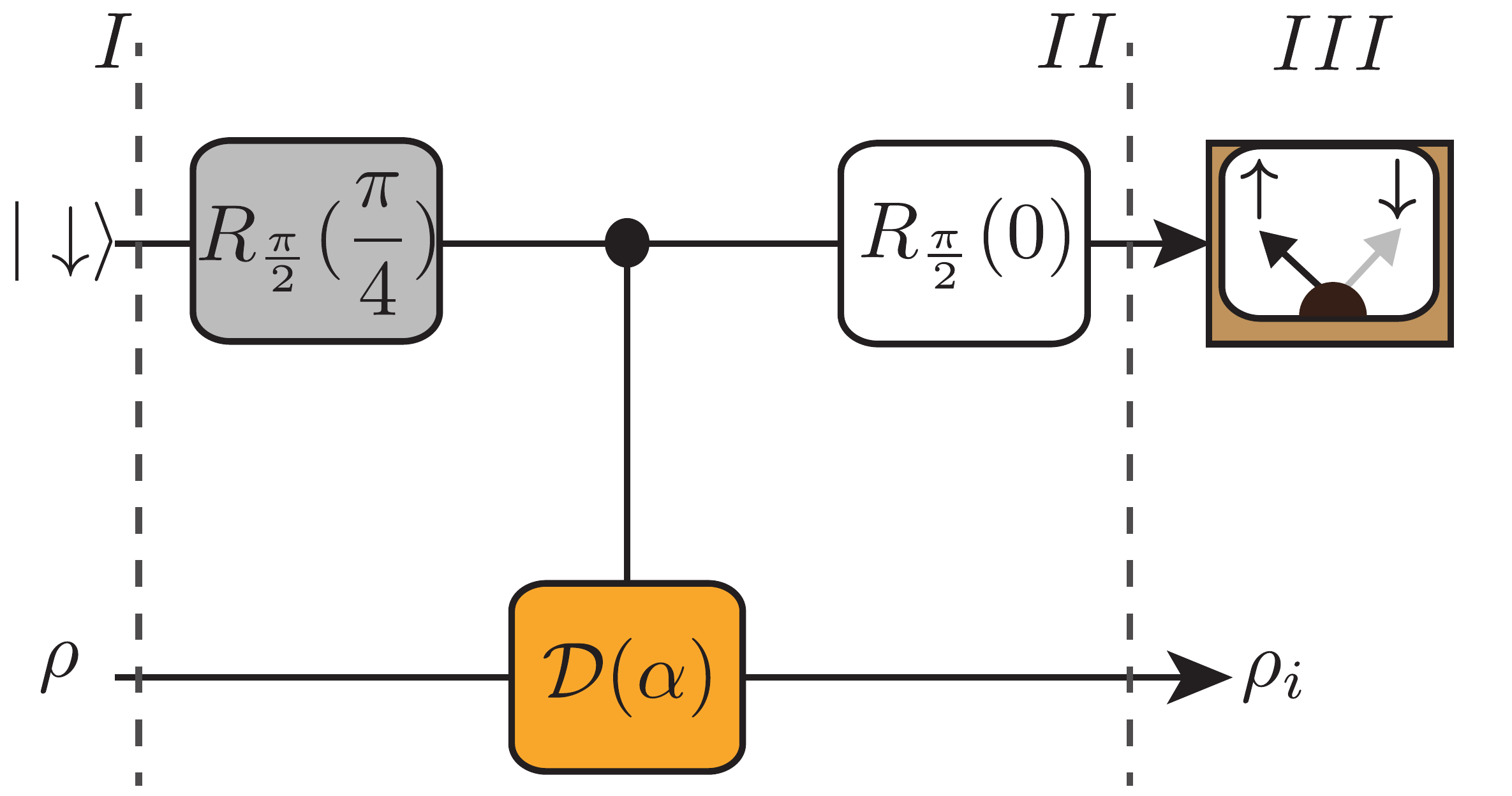}
\end{center}
\caption{A quantum circuit representing a Ramsey-type cycle for measuring HW observables via an ancillary readout qubit. The scheme consists of three steps. See text for detail. }
\label{Ramsey}
\end{figure}

\emph{Measuring HW observables}. One of the main advantages of the HW-observables is the fact that it is possible to measure their mean values via dichotomic  statistics of the measurements performed on a single qubit ancilla, coupled to the system of interest \cite{MeanCVqubit, Knill}. This is favored in some cases where qubit state read out and sequential measurements are efficient \cite{AsadianLG}. In this method the concept of generalized (or sometimes called indirect) measurement procedure is used. This type of measurement generally involves three elements. The
first element is the system of interest $S$, here the $d$-dimensional system, from which we want to gain
knowledge about its HW observables. The second element is a probe system, here an ancillary
read-out qubit. We let the two elements interact for a certain period of time and thereby 
correlations between $S$ and qubit are built up. A measurement apparatus, the
third element, is used to read out the qubit's state at the end of the interaction time. The
aim of this scheme is to obtain information on the state of $S$ via the state of the probe,
encoding the necessary information about $S$. 
This strategy can be applied using a Ramsey-type measurement scheme, a sequence of three essential steps (see fig.~\ref{Ramsey}):

\noindent
(I). Qubit and $S$ are initialized as  $\rho\otimes\proj{\downarrow}$. 

\noindent
(II). A fast rotation with suitable angles, $R_{\pi/2}(\pi/4)$, is performed on the qubit. Here $R_{\pi/2}(\varphi)$ denote $\pi/2$-rotations of a qubit with an adjustable phase $\varphi$. In the basis $\{ |\downarrow\rangle, |\uparrow\rangle\}$ they are defined as 
\begin{equation}
R_{\pi/2}(\varphi)=\left(\begin{array}{cc}  1 & e^{i\varphi} \\  -e^{-i\varphi} & 1\end{array}\right) \,.
\end{equation}
Then the qubit interacts with $S$ for a period of time generating a conditional displacement $U_\alpha=\mathbbm{1}\otimes\proj{\downarrow}+\mathcal{D}(\alpha)\otimes\proj{\uparrow}$ followed by the rotation $R_{\pi/2}(0)$. Thus the resulting evolution of the composite system is $U_{RM}=R_{\pi/2}(0)U_\alpha R_{\pi/2}(\pi/4)$.

\noindent
(III). A projective measurement on the probe qubit in the basis $\{\ket{\uparrow}, \ket{\downarrow}\}$ is performed. 
As a result of the measurement back action, the state of $S$ is projected onto one of the conditioned 
states $\rho_i$ depending on the outcome of the measurement. 

It is worth remarking that the implementation of the associated positive operator-valued measure (POVM) elements is illustrated by Neumark's theorem \cite{Naimark}, where the mean values of HW observables are obtained via two-valued POVM elements $E_j= M_j^\dag M_j=\bra{j}U_{RM}\ket{\downarrow}$, with $j\in\{\uparrow, \downarrow\}$ and $E_\uparrow+E_\downarrow=\mathbbm{1}$. These POVM elements are realized by the qubit measurement with corresponding probability $p_j=\text{Tr}\{\rho E_j\}$. 
Therefore we have,
\begin{equation}
\mean{Z}\equiv p_\uparrow-p_\downarrow=\text{Tr}\{\rho E_\uparrow\}-\text{Tr}\{\rho E_\downarrow\}=\dfrac{\text{Tr}\{\rho\mathcal{Q}( \alpha)\}}{\sqrt{2}}.
\end{equation}
This is a simple procedure in which the mean value of a HW observable is equal to the mean value of the qubit observable $Z$, up to a proportionality constant $\sqrt{2}$.
More importantly, the scheme with required length of  displacement amplitudes has already been implemented very efficiently in recent trapped ion experiments \cite{home}.

In hybrid system settings this approach offers an alternative procedure for experimental reconstruction of the quantum state via measuring an ancillary qubit system where the required interaction to implement the conditional displacement is of the simple form $( \lambda^* a +\lambda a^\dag )\proj{\uparrow}$, describing a linear coupling between the field quadrature and the qubit state with coupling strength $\lambda$ \cite{Merzbacher, Monroe}.  The other well-known scheme is based on displaced parity measurements  \cite{Davidovich}. In this approach also a Ramsey-type measurement is performed on an ancillary two-level atom coupled to a field. But it is experimentally more expensive, because in addition to the aforementioned coupling needed for displacing the field's state, a dispersive interaction of the form $g a^\dag a\proj{\uparrow}$ is needed to implement the parity measurement. This type of coupling is usually hard to realize specifically 
for more massive systems such as nano-mechanical resonators. Therefore, state estimation via HW observables offers an experimentally accessible procedure to a wider range of experimental setups. It is worth adding that, alternatively, continuous HWOs can be measured via homodyne measurements of the field quadratures by adapting the Vogel and Risken scheme \cite{Risken, Homodyne}.


\emph{Conclusions}. In the present work we put forward a Hermitian basis acting on $d$ dimensional Hilbert space. This basis holds distinctive features from those of GGM basis making it particularly useful in certain applications such as entanglement detection and tomography of high-dimensional systems. We give a thorough characterization of a number of relevant properties of this  basis, including the spectrum and (anti)-commutativity properties. Together with a theorem bounding the expectation value of general anti-commuting observables we furthermore demonstrate the usefulness of the HW-basis in entanglement detection by deriving a general method to obtain linear entanglement witnesses. Interesting roads for the future would be to investigate whether the set of observables presented in this work can simplify the analysis stabilizer systems defined in terms of HW operators \cite{kontext}. And since our construction of Bloch vectors is closely related to the phase-space picture it would be interesting to investigate applications for significant problems where phase space formulations are exploited, such as finding magic states \cite{magic}, determining the Wigner function \cite{Wootters,wigner} or the problem of finding symmetric informationally complete POVM \cite{sicpovm}.

\textit{Acknowledgements} We are grateful to S. Altenburg, F. Bernards, C. Budroni, D. Gross, F. Steinhoff, R. Werner, A. Winter and H. Zhu for productive discussions. We are furthermore very grateful to Otfried G\"uhne for showing us a proof of the fact that there cannot be more than three anti-commuting HW-observables. A.A. acknowledges the support
by Erwin Schr\"odinger Stipendium No. J3653-N27. MH and PE were supported by the European Commission (STREP ``RAQUEL''), CK by the ERC 267386 (``IRQUAT'') and MH, PE and CK by the Spanish MINECO, projects FIS2008-01236 and FIS2013-40627-P, with the support of FEDER funds, and by the Generalitat de Catalunya CIRIT, project 2014-SGR-966. PE furthermore acknowledges funding from the Swiss National Science Foundation (SNF) and the National Centres of Competence in Research "Quantum Science and Technology" (QSIT). MH furthermore acknowledges funding from the Juan de la Cierva fellowship (JCI 2012-14155), the Swiss National Science Foundation (AMBIZIONE PZ00P2$\_$161351) and the Austrian Science Fund (FWF) through the START project Y879-N27.

\bibliography{Bloch}
\begin{widetext}
\section{Appendix}

The appendix provides detailed derivations of the results from the main text including an explicit matrix representation (for d=3 and d=4) of the HW observables. Furthermore we give a detailed proof for Theorem 1 in the main text.

\subsection{I. Definition of the observables}
\label{sec:Ansatz}
We want to find $\chi$ defined in (5) of the main text such that 
\begin{equation}
\text{Tr}\{\mathcal{Q}(l,m)\mathcal{Q}(l',m')\} =d \delta_{l,l'}\delta_{m,m'}.
\end{equation}
Therefore we expand 
\begin{align*}
\label{eq:orth}
&\text{Tr}\{\mathcal{Q}(l,m)\mathcal{Q}(l',m')\} \\
&=|\chi|^2\text{Tr}\{\mathcal{D}(l,m)\mathcal{D}^\dag(l',m')\}+|\chi|^2\text{Tr}\{\mathcal{D}^\dag(l,m)\mathcal{D}(l',m')\} \\
& \ \  \ \ \ \  \chi^2\text{Tr}\{\mathcal{D}(l,m)\mathcal{D}(l',m')\}+{\chi^*}^2\text{Tr}\{\mathcal{D}^\dag(l,m)\mathcal{D}^\dag(l',m')\} \\
&=2d|\chi|^2\delta_{l,l'}\delta_{m,m'} \\
&+d e^{i\pi(l-m)} (\chi^2 e^{-i\pi d}+{\chi^*}^2 e^{i\pi d})\delta_{l+l',d}\delta_{m+m',d}=d \delta_{l,l'}\delta_{m,m'}.
\end{align*}
In order to satisfy the orthogonality condition we need to have that
\begin{equation}
|\chi|^2=\dfrac{1}{2} \ \ \ \ \text{and} \ \ \  \chi^2=- {\chi^*}^2 .
\end{equation}
Simple algebra leads to the solution $\chi=(1\pm i)/2$.

\subsection{II. Maximum eigenvalue of an HW observable}
The expectation value of a HW observable squared is given by
\begin{align*}
\sqrt{\mean{\mathcal{Q}^2(l,m)}}=\sqrt{1+\Im \mean{\mathcal{D}(2l,2m)}}.
\end{align*}
As the displacement operator is a unitary operator its eigenvalue is bounded by 1 and thus we need to take the imaginary part of the displacement operator. Therefore we have that
\begin{equation}
\sqrt{\mean{\mathcal{Q}^2(l,m)}}\leq |q_{\rm max}| \leq \sqrt{ 2}. \label{qmaxev}
\end{equation}
where $q_{\rm max}$ is the maximum eigenvalue.
In the special case $d=2$, $\Im \mathcal{D}(2l,2m)=0$, which simply the case where the HW observables reduce to Pauli operators. For higher dimensions however $\Im \mathcal{D}(2l,2m)$ is non-zero and  the maximum eigenvalue is bounded by $\sqrt{2}$. The absolute values of the eigenvalues for $\sqrt{\mathcal{Q}^2(l,m)}$ for all HW observables are
\begin{equation}
|q_n|=\sqrt{1+\sin\dfrac{4\pi n}{d}}
\label{HWspectrum}
\end{equation}
for $n=0, \cdots, d-1$.

\subsection{III. Explicit representation}
For the convenience of the reader here an explicit matrix representation of the HW observables is given. While for higher dimension the compact notation $ \alpha=\sqrt{\dfrac{\pi}{d}}(m+il)$ introduced in the main text becomes increasingly handy, it might seem unpractical in low dimension. We therefore stick to the notation provided in eq. (5) of the main text for presenting the matrix representation in d=3 and d=4.

\subsubsection{d=3}
For  the sake of a compact notation let $\chi=(1+ i)/2$ and $\omega = e^{\frac{2 i \pi}{3}}$.
\begin{align}
&\mathcal{Q}\text{(0,0)=} \begin{pmatrix}
 1 & 0 & 0 \\
 0 & 1 & 0 \\
 0 & 0 & 1 \\
\end{pmatrix}
\quad
\mathcal{Q}\text{(0,1)=} \begin{pmatrix}
0 & \chi^{*}& \chi\\
 \chi& 0 & \chi^{*}\\
 \chi^{*}& \chi& 0 \\
\end{pmatrix}
\quad
\mathcal{Q}\text{(0,2)=} \begin{pmatrix}
 0 & \chi& \chi^{*}\\
 \chi^{*}& 0 & \chi\\
 \chi& \chi^{*}& 0 \\
 \end{pmatrix}\nonumber\\
&\mathcal{Q}\text{(1,0)=}  \begin{pmatrix}
 \chi + \chi^{*}& 0 & 0 \\
 0 &  \chi \omega + \chi^{*}\omega^{*} & 0 \\
 0 & 0 &\chi\omega^{*}+\chi^{*}\omega \\
\end{pmatrix}
\quad
\mathcal{Q}\text{(1,1)=}  \begin{pmatrix}
 0 &-\chi^{*}\omega& - \chi \omega \\
-\chi\omega^{*}& 0 & - \chi^{*}\\
  -\chi^{*}\omega^{*}& -\chi & 0 \\
\end{pmatrix}\nonumber\\
&\mathcal{Q}\text{(1,2)=}  \begin{pmatrix}
  0 & \chi\omega^{*} &\chi^{*}\omega^{*}\\
 \chi^{*}\omega& 0 &\chi \\
\chi \omega  & \chi^{*} & 0 \\
\end{pmatrix}
\quad
\mathcal{Q}\text{(2,0)=}  \begin{pmatrix}
 \chi + \chi^{*}& 0 & 0 \\
 0 & \chi\omega^{*}+\chi^{*}\omega  & 0 \\
 0 & 0 &   \chi \omega + \chi^{*}\omega^{*} \\
\end{pmatrix}\nonumber\\
&\mathcal{Q}\text{(2,1)=}\begin{pmatrix}
 0 &\chi^{*}\omega^{*}& \chi\omega^{*}  \\
\chi \omega & 0 & \chi^{*}\\
\chi^{*}\omega& \chi & 0 \\
\end{pmatrix}
\mathcal{Q}\text{(2,2)=} \begin{pmatrix}
 0 &  \chi \omega & \chi^{*}\omega\\
 \chi^{*}\omega^{*}& 0 & \chi \\
  \chi \omega^{*} &\chi^{*} & 0 \\
\end{pmatrix}
\end{align}

\subsubsection{d=4}

\begin{align}
&\mathcal{Q}\text{(0,0)=} \begin{pmatrix}
 1 & 0 & 0 & 0 \\
 0 & 1 & 0 & 0 \\
 0 & 0 & 1 & 0 \\
 0 & 0 & 0 & 1 \\
\end{pmatrix}
\quad
\mathcal{Q}\text{(0,1)=}\begin{pmatrix}
 0 &\chi^{*} & 0 &\chi \\
\chi & 0 &\chi^{*} & 0 \\
 0 &\chi & 0 &\chi^{*} \\
\chi^{*} & 0 &\chi & 0 \\
\end{pmatrix}
\quad
\mathcal{Q}\text{(0,2)=} \begin{pmatrix}
 0 & 0 & 1 & 0 \\
 0 & 0 & 0 & 1 \\
 1 & 0 & 0 & 0 \\
 0 & 1 & 0 & 0 \\
\end{pmatrix}\nonumber\\
&\mathcal{Q}\text{(0,3)=}\begin{pmatrix}
 0 &\chi & 0 &\chi^{*} \\
\chi^{*} & 0 &\chi & 0 \\
 0 &\chi^{*} & 0 &\chi \\
\chi & 0 &\chi^{*} & 0 \\
\end{pmatrix}
\quad
\mathcal{Q}\text{(1,0)=} \begin{pmatrix}
 1 & 0 & 0 & 0 \\
 0 & -1 & 0 & 0 \\
 0 & 0 & -1 & 0 \\
 0 & 0 & 0 & 1 \\
\end{pmatrix}
\quad
\mathcal{Q}\text{(1,1)=}\frac{1}{\sqrt{2}} \begin{pmatrix}
 0 & -i& 0 & 1\\
 i& 0 & -1& 0 \\
 0 & -1& 0 & i\\
 1& 0 & -i& 0 \\
\end{pmatrix}\nonumber\\
& \mathcal{Q}\text{(1,2)=} \begin{pmatrix}
 0 & 0 & -i & 0 \\
 0 & 0 & 0 & i \\
 i & 0 & 0 & 0 \\
 0 & -i & 0 & 0 \\
\end{pmatrix}
\quad
\mathcal{Q}\text{(1,3)=}\frac{1}{\sqrt{2}} \begin{pmatrix}
 0 & -i& 0 & -1\\
 i& 0 & 1& 0 \\
 0 & 1& 0 & i\\
 -1& 0 & -i& 0 \\
\end{pmatrix}
\quad
\mathcal{Q}\text{(2,0)=} \begin{pmatrix}
 1 & 0 & 0 & 0 \\
 0 & -1 & 0 & 0 \\
 0 & 0 & 1 & 0 \\
 0 & 0 & 0 & -1 \\
\end{pmatrix}\nonumber\\
&\mathcal{Q}\text{(2,1)=}  \begin{pmatrix}
 0 & -\chi & 0 &\chi^{*} \\
 -\chi^{*} & 0 &\chi & 0 \\
 0 &\chi^{*} & 0 & -\chi \\
\chi & 0 & -\chi^{*} & 0 \\
\end{pmatrix}
\quad
\mathcal{Q}\text{(2,2)=} \begin{pmatrix}
 0 & 0 & -1 & 0 \\
 0 & 0 & 0 & 1 \\
 -1 & 0 & 0 & 0 \\
 0 & 1 & 0 & 0 \\
\end{pmatrix}
\quad
\mathcal{Q}\text{(2,3)=} \begin{pmatrix}
 0 & -\chi^{*} & 0 &\chi \\
 -\chi & 0 &\chi^{*} & 0 \\
 0 &\chi & 0 & -\chi^{*} \\
\chi^{*} & 0 & -\chi & 0 \\
\end{pmatrix}\nonumber\\
&\mathcal{Q}\text{(3,0)=} \begin{pmatrix}
 1 & 0 & 0 & 0 \\
 0 & 1 & 0 & 0 \\
 0 & 0 & -1 & 0 \\
 0 & 0 & 0 & -1 \\
\end{pmatrix}
\quad
\mathcal{Q}\text{(3,1)=}\frac{1}{\sqrt{2}} \begin{pmatrix}
 0 & -1& 0 & -i\\
 -1& 0 & -i& 0 \\
 0 & i& 0 & 1\\
 i& 0 & 1& 0 \\
\end{pmatrix}
\quad
\mathcal{Q}\text{(3,2)=} \begin{pmatrix}
 0 & 0 & i & 0 \\
 0 & 0 & 0 & i \\
 -i & 0 & 0 & 0 \\
 0 & -i & 0 & 0 \\
\end{pmatrix}\nonumber\\
&\mathcal{Q}\text{(3,3)=} \frac{1}{\sqrt{2}} \begin{pmatrix}
 0 & 1& 0 & -i\\
 1& 0 & -i& 0 \\
 0 & i& 0 & -1\\
 i& 0 & -1& 0 \\
\end{pmatrix}
\end{align}

\subsection{IV. Anticommutativity Bound}
\setcounter{theorem}{0}
\begin{theorem}
Let $\{\lambda_{i}\}_{i \in \mathcal{I}}$ with the index set $\mathcal{I}=\{1,2,\ldots,d^{2}\}$ denote an orthonormal self-adjoint basis $\mathcal{B}$ of a $d$-dimensional Hilbert space $\mathcal{H}$ and $\mathcal{A} \subseteq \mathcal {I}$ refer to a subset of $\mathcal{B}$ such that $\frac{1}{2}\sqrt{\sum_{i \neq j \in \mathcal{A}}\langle\{\lambda_i,\lambda_j\}\rangle^2}  \leq \mathcal{K}$. 
Then the corresponding Bloch vector components $c_{i}$ of any density matrix $\rho \in \mathcal{H}$ expressed in $\mathcal{B}$ as $\rho=\sum_{i \in \mathcal{A}}c_{i} \lambda_{i} + \sum_{l \in \overline{\mathcal{A}}} c_{l} \lambda_{l}$ can be bounded by
\begin{equation}
\sum_{i \in \mathcal{A}} c_{i}^{2} \leq \frac{\max_{i \in \mathcal{A}} \langle \lambda_{i}^{2} \rangle  +\mathcal{K} }{(\min_{i \in \mathcal{A}} \Tr(\lambda_{i}^{2}))^{2}}
\end{equation}
\end{theorem}

\begin{proof}
Any $\rho \in \mathcal{H}$ can be expressed in an orthonormal basis $\mathcal{B}$. The chosen basis $\mathcal{B}$ can always be divided into $\mathcal{A} \cup \overline{\mathcal{A}}$, where the set $\mathcal{A}$ is at worst trivial.

Now consider the observable $\mathcal{O}$
\begin{equation}
\mathcal{O}:=\sum_{i \in \mathcal{A}} c_{i} \lambda_{i}.
\end{equation} 
We trivially have
\begin{equation}
(\Delta \mathcal{O} )^{2} := \langle \mathcal{O}^{2} \rangle - \langle \mathcal{O} \rangle^{2} \geq 0 
\end{equation}
since the variance of any observable is positive.

Inserting for $\mathcal{O}$ we obtain
\begin{align}
\langle \mathcal{O}^{2} \rangle = \langle \sum_{i \in \mathcal{A}} c_{i}^{2} \lambda_{i}^{2} + \sum_{i \neq j \in \mathcal{A}} c_{i} c_{j} \lambda_{i} \lambda_{j} \rangle=\langle \sum_{i \in \mathcal{A}} c_{i}^{2} \lambda_{i}^{2} +\frac{1}{2}\sum_{i \neq j \in \mathcal{A}} c_{i} c_{j} \underbrace{(\lambda_{i} \lambda_{j}+\lambda_{j} \lambda_{i})}_{ \{\lambda_{i}, \lambda_{j}\}} \rangle \leq \nonumber\\ \sum_{i \in \mathcal{A}} c_{i}^{2} \langle \lambda_{i}^{2} \rangle +  \underbrace{\sqrt{\sum_{i \neq j \in \mathcal{A}}c_{i}^2 c_{j}^2}}_{\leq\sum_{i}c_i^2}\underbrace{\frac{1}{2}\sqrt{\sum_{i \neq j \in \mathcal{A}}\langle\{\lambda_i,\lambda_j\}\rangle^2}}_{ \leq \mathcal{K}},
\end{align}
with application of the boundedness of the anti commutator of $\mathcal{A}$, as well as
\begin{equation}
\langle \mathcal{O} \rangle^{2} = (\sum_{i \in \mathcal{A}} c_{i} \Tr(\rho \lambda_{i}))^{2} = (\sum_{j \in \mathcal{A}} \Tr(\sum_{i \in \mathcal{A}} c_{i} \lambda_{i}) c_{j} \lambda_{j}) + \sum_{l \in \overline{\mathcal{A}}} \Tr(\sum_{i \in \mathcal{A}} c_{i} \lambda_{i}) c_{l} \lambda_{l}))^{2}= (\sum_{i \in \mathcal{A}} c_{i}^{2} \Tr(\lambda_{i}^{2}))^{2}
\end{equation}
by use of orthonormality of $\mathcal{B}$. Summing up the above three equations yields
\begin{equation}
0 \leq\langle \mathcal{O}^{2} \rangle - \langle \mathcal{O} \rangle^{2} \leq \sum_{i \in \mathcal{A}} (c_{i}^{2} \langle \lambda_{i}^{2} \rangle +  c_{i}^2 \mathcal{K}) - (\sum_{i \in \mathcal{A}} c_{i}^{2} \Tr(\lambda_{i}^{2}))^{2} \leq \sum_{i \in \mathcal{A}} c_{i}^{2} \Big(\max_{i\in \mathcal{A}}\langle \lambda_{i}^{2} \rangle - (\sum_{i \in \mathcal{A}} c_{i}^{2}) \min_{i\in \mathcal{A}}\Tr(\lambda_{i}^{2})^{2}\Big) + \sum_{i \in \mathcal{A}} c_{i}^2\mathcal{K}.
\end{equation}
thus
\begin{equation}
\sum_{i \in \mathcal{A}} c_{i}^{2} \Big(\max_{i\in \mathcal{A}}\langle \lambda_{i}^{2} \rangle - (\sum_{i \in \mathcal{A}} c_{i}^{2}) \min_{i\in \mathcal{A}}\Tr(\lambda_{i}^{2})^{2}\Big) + \sum_{i \in \mathcal{A}} c_{i}^2\mathcal{K}
\leq
\sum_{i \in \mathcal{A}} c_{i}^{2} \Big(\max_{i\in \mathcal{A}}\langle \lambda_{i}^{2} \rangle - (\sum_{i \in \mathcal{A}} c_{i}^{2}) \min_{i\in \mathcal{A}}\Tr(\lambda_{i}^{2})^{2}  \Big)
\end{equation}
The positivity enforces
\begin{equation}
0 \leq \max_{i\in \mathcal{A}}\langle \lambda_{i}^{2} \rangle - (\sum_{i \in \mathcal{A}} c_{i}^{2}) \min_{i\in \mathcal{A}}\Tr(\lambda_{i}^{2})^{2} +  \mathcal{K} \Longleftrightarrow \sum_{i \in \mathcal{A}} c_{i}^{2} \leq \frac{\max_{i\in \mathcal{A}}\langle \lambda_{i}^{2} \rangle +\mathcal{K}}{\min_{i\in \mathcal{A}}\Tr(\lambda_{i}^{2})^{2}}
\end{equation}
\end{proof}
As it may still be cumbersome to compute $\mathcal{K}$ we point out that $ \mathcal{K}=\frac{1}{2}\sqrt{\sum_{i \neq j \in \mathcal{A}}\langle\{\lambda_i,\lambda_j\}\rangle^2}\leq\frac{1}{2}\sqrt{\sum_{i \neq j \in \mathcal{A}}\left\|\{\lambda_i,\lambda_j\}\right\|^{2}_{\infty}}$. As the HW-Basis obeys $\Tr(\lambda_{i}^{2})^{2}=d^2\,\forall\,\lambda_i$ and the coefficients $c_i$ are related to the Bloch vector components by $c_i=\frac{1}{d}\langle \mathcal{Q}(\alpha)\rangle$ we can conclude with the simple bound
\begin{equation}
\label{k}
\sum_{i \in \mathcal{A}} \langle \mathcal{Q}(\alpha_i)\rangle^{2} \leq 1+ \max_n \left( \sin\dfrac{4\pi n}{d}\right) + \frac{1}{2}\sqrt{\sum_{i \neq j \in \mathcal{A}} \left\|\{\mathcal{Q}(\alpha_i),\mathcal{Q}(\alpha_j)\}\right\|^{2}_{\infty}}
\end{equation}

\subsubsection{A. Entanglement detection examples}
\label{sec:example}

\emph{GHZ(3,4)}.--- Let us turn to tripartite systems of dimension four ("`ququarts"') with the GHZ state defined as $\rho_{GHZ(3,4)}=\frac{1}{4}\sum_{i,j=0}^{3} |iii\rangle\langle jjj|$.
First we choose a set of indices according to ensure anti-commutativity along a partition. For bipartite entanglement detection, we consider the partition into the subsystems $A$ and $BC$. Using the abbreviations for the three measurement settings with amplitudes $\alpha_1=\sqrt{\pi/4}$, $\alpha_2=\sqrt{\pi}i$ and $\alpha_3=\alpha_1+\alpha_2$ obeying (14) from the main text.  The following values for the observables $\mathcal{Q}_{A}(\alpha_1) \otimes \mathbbm{1} \otimes \mathcal{Q}_{C}(\alpha_1)$, $\mathcal{Q}_{A}(\alpha_3) \otimes \mathcal{Q}_{B}(\alpha_2) \otimes \mathcal{Q}_{C}(\alpha_3)$ and $\mathcal{Q}_{A}(\alpha_2) \otimes \mathcal{Q}_{B}(\alpha_2) \otimes \mathcal{Q}_{C}(\alpha_2)$ have the expectation values $1,-1,1$. Note that subsequently the first cut $\mathcal{Q}_{A}$ is anti-commuting as well as $\mathcal{Q}_{B}$ is anti-commuting, while is $\mathcal{Q}_{C}$ commuting, resulting in the second cut $\mathcal{Q}_{B} \otimes \mathcal{Q}_{C}$ being anti-commuting as required. In this particular case 
the maximal eigenvalue used in (\ref{HWspectrum}) is equal to one, thus
\begin{align}
 \nonumber
|\langle \mathcal{Q}_{A}(\alpha_1)\otimes  \mathbbm{1}\otimes \mathcal{Q}_{C}(\alpha_1) \rangle|
+|\langle \mathcal{Q}_{A}(\alpha_3)\otimes \mathcal{Q}_{B}(\alpha_2) \otimes \mathcal{Q}_{C}(\alpha_3) \rangle|
+|\langle \mathcal{Q}_{A}(\alpha_2)\otimes \mathcal{Q}_{B}(\alpha_2)\otimes \mathcal{Q}_{C}(\alpha_2)\rangle|=3\underbrace{\leq}_{SEP}1,
\end{align}
which clearly violates the obtained inequality. If the state would have been separable into $A$ and $BC$, we could have split the commuting tensor products of $ABC$ into the anti commuting parts $A$ and $BC$ for whose the theorem should hold, as well as the bound above following from it.
The entanglement detection is achieved with only three measurements and a very satisfactory noise resistance, allowing us for $\rho_{Noise}=p\rho_{GHZ(3,4)}+\frac{(1-p)}{64}\mathbbm{1}$ to choose as little as $p > \frac{1}{3}$ to prove entanglement.
We point out that many choices of observables would have been possible, although some may lead to less noise resistance. The convenient convexity of the moduli also allows a trivial extension to detect genuine multipartite entanglement. The symmetrized extension of the above observables yields seven elements, all of which have expectation values of either $1$ or $-1$. The straightforward application of the outlined procedure yields a bound of $3$ for biseparable states and thus detection of genuine multipartite entanglement with a noise resistance of $p=\frac{4}{7}$.

\emph{Maximally entangled state in d=9}.--- For systems of dimension nine with the maximally entangled state defined as $\rho_{max}=\frac{1}{9}\sum_{i,j=0}^{8} |ii\rangle\langle jj|$ and the following abbreviations for the five measurement settings with amplitudes $\alpha_1=\sqrt{\pi}i/3$, $\alpha_2=\sqrt{\pi}8i/3$ , $\alpha_3=\sqrt{\pi}4/3$, $\alpha_4=\alpha_1 + \alpha_3$ and $\alpha_5=\alpha_2 + \alpha_3$. We make use of the following three observables $\mathcal{Q}_{A}(\alpha_1) \otimes \mathcal{Q}_{B}(\alpha_2)$, $\mathcal{Q}_{A}(\alpha_3) \otimes \mathcal{Q}_{B}(\alpha_3) $ and $\mathcal{Q}_{A}(\alpha_4) \otimes \mathcal{Q}_{B}(\alpha_5) $ that have the expectation values $1,1,1$.
As for odd dimension there are no HW observables that anticommute (i.e. $k \neq 0$) we use eq. \ref{k} and thus
\begin{align}
|\langle \mathcal{Q}_{A}(\alpha_1)\otimes \mathcal{Q}_{B}(\alpha_2) \rangle|
+|\langle \mathcal{Q}_{A}(\alpha_3)\otimes \mathcal{Q}_{B}(\alpha_3)  \rangle|
+|\langle \mathcal{Q}_{A}(\alpha_4)\otimes \mathcal{Q}_{B}(\alpha_5)\rangle|=3\underbrace{\leq}_{SEP} 2.41987
\end{align}
which again violates the obtained inequality.

\subsubsection{B. An upper bound on the number of exactly anticommuting observables}
In this section, we show that using the phase space-approach one 
can generate at most sets of three observables, which are mutually
anti commuting. In other words, if we consider a set of four observables, 
then they can not be pairwise anti commuting. 

To see this, consider a set of four observables, which are pairwise 
anti commuting. We parameterize the four observables by two-dimensional
real vectors $\vec{A},\vec{B},
\vec{C}$ and $\vec{D}$, describing the displacement in phase space.
The anticommutativity leads to six conditions on the vectors, the
first one reading
\begin{equation}
|\vec{A} \times \vec{B}|= k_1 \frac{\pi}{2} 
\end{equation}
with $k_1$ being an odd integer. The other five conditions are analogous, 
with odd numbers $k_2, \dots, k_6$. Solving now the first five equations
for $\vec{A},\vec{B}, \vec{C}$ and $\vec{D}$ and inserting the resulting
conditions into the sixth one, leads to the insight that the $k_i$ have to
obey:
\begin{equation}
k_1 k_6 + k_2 k_5 - k_3 k_4 = 0.
\end{equation}
This condition, however, can not be fulfilled if all the $k_i$ are odd.

\end{widetext}

\end{document}